\documentclass{article}
\usepackage[utf8]{inputenc}

\usepackage{etoolbox}
\newtoggle{full_version}
\toggletrue{full_version}

\usepackage{epsfig}
\usepackage{amsmath}
\usepackage{amssymb}
\usepackage{amsthm}
\usepackage{color}
\usepackage{url}
\usepackage{algorithm}
\usepackage{enumitem}

\usepackage{graphicx}
\usepackage{subcaption}

\usepackage{cite}

\newlength{\figurewidth}
\newlength{\smallfigurewidth}

\def\bfx{{ \mathbf{x}}}

\def\defeq{{\stackrel{\Delta}{=}}}

\theoremstyle{plain}
\newtheorem{theorem}{Theorem}
\newtheorem{proposition}[theorem]{Proposition}
\newtheorem{lemma}[theorem]{Lemma}

\theoremstyle{definition}
\newtheorem{definition}[theorem]{Definition}

\theoremstyle{remark}

\newcommand\blfootnote[1]{%
  \begingroup
  \renewcommand\thefootnote{}\footnote{#1}%
  \addtocounter{footnote}{-1}%
  \endgroup
}

\DeclareMathOperator{\sign}{sign}

\begin{document}

\title
{\LARGE
\textbf{Correct Convergence of Min-Sum \\Loopy Belief Propagation\\ in a Block Interpolation Problem}
}


\author{%
Yutong Wang, Matthew G. Reyes, David L. Neuhoff
}

\maketitle


\begin{abstract}
This work proves a new result on the correct convergence of Min-Sum Loopy Belief Propagation (LBP) in an interpolation problem on a square grid graph. The focus is on the notion of local solutions, a numerical quantity attached to each site of the graph that can be used for obtaining MAP estimates.
The main result is that over an $N\times N$ grid graph with a one-run boundary configuration, the local solutions at each $i \in B$ can be calculated using Min-Sum LBP by passing difference messages in $2N$ iterations, which parallels the well-known convergence time in trees.
\end{abstract}

\section{Introduction}\label{sec:introduction}
\iftoggle{full_version}{\blfootnote{An abbreviated version of this paper has been submitted to ISIT 2017.}
\blfootnote{Authors addresses: Yutong Wang, Matthew G. Reyes, David L. Neuhoff, EECS Dept., University of Michigan; email: \{yutongw, mgreyes, neuhoff\}@umich.edu.}
}
This paper demonstrates the correct convergence of Loopy Belief Propagation (LBP) in the MAP interpolation of a block of sites given a configuration on its boundary, in the context of a uniform Ising Markov random field. There has been considerable work in analyzing the performance of LBP in the context of maximization problems, for example \cite{weiss2000correctness,weiss2001optimality, wainwright2004tree, bayati2008max}. This paper presents both a new problem setting for Belief Propagation (BP) and a new method of analysis.

In the context of Markov models, a very natural setting is that of MAP estimation of a subset of sites conditioned on a configuration on its boundary, as the Markov property itself tells us that a subset of sites is conditionally independent of all other sites if we know the configuration on the subset's boundary. Markov models are often expressed as products of functions on single nodes and edges of the associated Markov graph. Thus by taking the negative logarithm of the probability, MAP estimation can be formulated as what is referred as a {\em min-sum} problem, that of finding configurations that minimize a sum of functions defined on single nodes and edges of the graph. Belief Propagation is a recursive distributed algorithm that can be applied to a min-sum problem.\footnote{Belief Propagation in the context of MAP estimation is more often studied as a max-product problem, the variant obtained without taking the negative logarithm.}  

The Markov model considered in this paper is a uniform Ising model with positive correlation on a square grid graph with edges connecting horizontally and vertially adjacent nodes and with nodes assigned values $+1$ (black) or $-1$ (white) \cite{baxter2007exactly}. This is a single-parameter binary model that favors configurations in which neighboring nodes have the same value. Edges on which the two endpoints have different values are called {\em odd bonds}. Our problem is to MAP estimate the configuration $\bfx_B$ on a subset $B$ of sites conditioned on a boundary configuration $\bfx_{\partial B}$. In this context, MAP estimation amounts to finding configurations that minimize the sum

\begin{equation}\label{eq:odd bonds}O(\bfx_B, \bfx_{\partial B}) = \sum\limits_{\{i,j\}:i\in B}\mathbb{I}_{(x_j \not= x_i)}\end{equation}

\noindent of odd bonds over all edges in the graph with at least one endpoint in $B$. This problem arose in the context of an image compression application modeling binary images as instances of such an Ising model \cite{reyes2014lossy} and also in the context of grayscale image reconstruction \cite{farmer2011cutset}. Analytical solutions for the set of MAP configurations conditioned on boundary configurations containing 2 or 4 odd bonds have been found \cite{reyes2011cutset,reyes2014lossy}. Such boundaries are termed, respectively, \emph{1-run} and \emph{2-run boundaries}. The MAP configurations on a block conditioned on its boundary are referred to as {\em global solutions} for the boundary. 

Min-Sum LBP is a popular distributed message-passing algorithm for minimizing a sum of functions defined on edges of a graph. As a distributed algorithm, it does not attempt to compute global solutions, but rather, for each site, the minimum numbers of odd bonds in configurations where site $i$ is black ($x_i=+1$) or white ($x_i=-1$),
\begin{equation}\label{white}O^*_i(\pm 1,\mathbf x_{\partial B}) :=  \min\limits_{\mathbf x_B: x_i=\pm 1}O(\mathbf x_B,\mathbf x_{\partial B}).
\end{equation}
These minimum numbers of odd bonds provide some information regarding the set of global solutions. For example, if $O^*_i(-1,\mathbf x_{\partial B}) < O^*_i(1,\mathbf x_{\partial B})$, then we can say that site $i$ has value -1 in all global solutions, whereas if $O^*_i(-1,\mathbf x_{\partial B}) > O^*_i(1,\mathbf x_{\partial B})$, site $i$ has value 1 in all global solutions. On the other hand, if $O^*_i(-1,\mathbf x_{\partial B}) = O^*_i(1,\mathbf x_{\partial B})$, what we can say is that there exists a global solution in which site $i$ has value -1, and there exists a global solution in which site $i$ has value 1. Moreover, as pointed out in \cite{weiss2001optimality} when there are multiple sites such that $O^*_i(-1,\mathbf x_{\partial B}) = O^*_i(1,\mathbf x_{\partial B})$, a joint configuration on these sites cannot be chosen independently of each other.

In practice, the messages are normalized to prevent numerical overflow. As a result, the goal of BP becomes computing the difference 
\begin{equation}\label{definition: local solutions}o^*_i(\mathbf x_{\partial B}) = O^*_i(-1,\mathbf x_{\partial B}) - O^*_i(1,\mathbf x_{\partial B}),
\end{equation}

\noindent which we refer to as the \emph{local solution} at site $i$ given boundary configuration $x_{\partial B}$. At the $n$-th iteration of message-passing, an estimate $\hat o^n_i(\mathbf x_{\partial B})$ of the local solution at site $i$ is produced. If $B$ were a tree, i.e., an acyclic graph, the usual argument of the correct convergence of BP on trees could be adapted here to show that $\hat o^n_i(\mathbf x_{\partial B})$ converges $o^*_i(\mathbf x_{\partial B})$. For cyclic graphs such as the grid graphs considered in the present paper, general convergence is unknown except in special cases such as when the graph is a single cycle \cite{weiss2000correctness}. However, it was observed in \cite{reyes2011cutset} that empirically LBP converged to the correct local solutions for a 1-run boundary.

It is our belief that LBP can be an effective distributed algorithm for the MAP interpolation problem posed here.  While a complete understanding of the correct convergence properties of LBP is currently beyond our means, in this paper we prove in Theorem \ref{main theorem} that it correctly converges in the case where $B$ is an $N\times N$ grid graph with horizontal and vertical edges with a one-run boundary configuration. Specifically, we show that the local solution at each $i \in B$ can be calculated using Min-Sum Belief Propagation by passing difference messages in $2N$ iterations. We define the \emph{Forward} and \emph{Backward Convergence Property}, which are crucial for our analysis of the convergence.  To verify the correctness of the converged results of LBP, we use Proposition \ref{correctness oracle}, which is proven by leveraging the results in \cite{reyes2014lossy}. Thus, the results of this paper demonstrate that at least in the case of one-run boundaries, LBP converges to the correct local solution in what amounts with a minimal number of iterations. We hope our work here gives some theoretical justification for using LBP and local solutions for interpolation problems beyond this setting.

The remainder of this paper is as organized as follows. In Section \ref{sec:background}, we introduce background on graphs, the boundary interpolation problem, and Belief Propagation. In Section \ref{sec:local}, we introduce our message recursion, local solutions, and state what the correct local solutions are. In Section \ref{sec:forbackstab}, we introduce and discuss the concepts of forward and backward convergence used to prove our results. \iftoggle{full_version}{In Section \ref{sec:proof of main theorem}, we present the proof of our main result, Theorem \ref{main theorem}.}{Due to lack of space, our results are stated without proof, though these are available in \cite{}[{\color{red}full arxiv paper}]. However, one interesting case of the proof of our main result, Theorem \ref{main theorem}, is included in Section \ref{sec:proof of main theorem}.}

%
%

\section{Background and Problem Formulation}\label{sec:background}
We introduce notation on graphs, configurations, etc. Let $\mathbb{N} = \{1,2,\dots\}$ and $\mathbb{N}_0 = \mathbb{N} \cup \{0\}$. Edges in an undirected graph are written as $\{i,j\}$. In a directed graph, edges are written as $j \to i$. For an undirected graph $G= (V,E)$ and a subset $S \subseteq V$, we let $\partial S := \{j \in V \setminus S \mid \exists i \in S:\, \{j,i\} \in E\}$. Abusing notation, we often refer to a subset $S\subseteq V$ as if it is a subgraph. For instance, the statement ``suppose $S$ is connected" means ``suppose the $G$-induced subgraph on $S$ is connected". Another abuse of notation is $\partial i := \partial \{i\}$.

\subsection{Grid graphs, configurations, and odd bonds}\label{sec:setting of the paper}

In this subsection, we define the setting that we work in for the majority of the paper. Let  $G=(V,E)$ be the $(N+2)\times (N+2)$ \emph{grid graph} with the {\em 4-neighbor topology} in which the sites are arranged in a square lattice and the edges consist of horizontally and vertically adjacent sites of $V$. Two sites connected by an edge are referred to as {\em neighbors}. The \emph{interior} is the set of sites $B \subseteq V$ having four neighbors. The set $\partial B$ is the {\em boundary} of $B$, i.e., $\partial B = V \setminus B$.

For each site $i\in V$, $x_i\in\{-1,1\}$ is an assignment to site $i$. An assignment to a set of sites $S$ is called a \emph{configuration} and denoted $\bfx_S$. For concreteness, $x_i = 1$ (resp.\ $x_i = -1$) means that site $i$ is colored black (resp.\ white). An edge $\{i,j\}$ with $x_i \not= x_j$ is called an {\em odd bond}.  A configuration $\bfx_{\partial B}$ on $\partial B$ is called a \emph{boundary configuration}. Finally, we define one-run boundaries:

\begin{definition}[\textbf{One-run boundary}]\label{def:one-run boundary}
Let $\bfx_{\partial B}$ be a boundary configuration. Define $R^+ = \{i \in \partial B \mid x_i = +1\}$ and define $R^-$ similarly. We say that $\bfx_{\partial B}$ is a \emph{one-run configuration} if $R^+$ is a connected subgraph. 
\end{definition}


\subsection{MAP Estimation and Global Solutions}

Given $\bfx_{\partial B}$ and interior configuration $\bfx_B$, the quantity $O(\bfx_B, \bfx_{\partial B})$ as defined by equation (\ref{eq:odd bonds}) is the number of odd bonds within $B$ and between $B$ and $\partial B$. The {\em global minimum} number of odd bonds between an interior configuration $\mathbf x_B$ and the given boundary configuration $\mathbf x_{\partial B}$ is
$$
O^*_B(\mathbf x_{\partial B}) :=  \min\limits_{\mathbf x_B}O(\mathbf x_B,\mathbf x_{\partial B}).$$

In \cite{reyes2011cutset, reyes2014lossy}, all MAP solutions were found for all one-run boundary configurations and at least one MAP solution was for found for every two-run boundary configuration. Specifically, for boundaries consisting of one-run of white and one-run of black, the MAP solutions consisted of configurations generated by a shortest path connecting the endpoints of either runs. In this work, we refer to a MAP solution as a \emph{global solution}.

\subsection{Belief Propagation}

We first review BP for interpolating from the boundary in the context of a tree. For now, let $G= (V,E)$ be an arbitrary tree apart from the grid graph currently under consideration. Let $B\subseteq V$ be a subtree and $\bfx_{\partial B}$ a boundary configuration. Consider any two adjacent nodes $i,j \in B$. Removing edge $\{i,j\}$ from $B$ disconnects $B$ into two connected components. Let $T_j^i\subseteq B$ be the connected component containing $j$. 
Define $M_{j\to i}(-1)$ to be the minimal number of odd bonds in $T^i_j \cup \partial B$ over all possible configurations $\bfx_{T^i_j}$ on $T^i_j$ plus the odd bond between $j$ and $i$, if any. More precisely,

$$M_{j\to i}(-1) = \mathbb{I}_{(x_j \ne 1)}+ \min_{\bfx_{T^i_j}} O(\bfx_{T^i_j}, \bfx_{\partial B}).$$

Recall $O^*_i(-1,\mathbf x_{\partial B})$ as defined by equation (\ref{white}). Below, we suppress the dependency on $\mathbf x_{\partial B}$ and simply write $O^*_i(-1)$. Likewise, we write $O^*_i(1)$. Since $B$ is a tree, it is easy to see that $O^*_i(-1)$ could be expressed as
\begin{equation}\label{unnormalized read out}
O^*_i(-1) =  \sum\limits_{j\in\partial i} M_{j\rightarrow i}(-1).
\end{equation}
We define $M_{j\rightarrow i} \defeq [M_{j\rightarrow i}(-1),M_{j\rightarrow i}(1)]$ to be the {\em message} from site $j$ to site $i$. Using a recursive argument, it is straightforward to show 
\begin{eqnarray}
M_{j\rightarrow i}(x_i) & = & \min \limits_{x_j \in \{\pm 1\}}\left\{ \mathbb{I}_{(x_i\not=x_j)} + \sum\limits_{k\in\partial j\setminus i}M_{k\rightarrow j}(x_j)\right\} \nonumber
\end{eqnarray}

The above recursion relation induces a message-passing algorithm:

\begin{definition}\label{unnormalized message} Given a boundary configuration $\bfx_{\partial B}$, for each edge $\{i,j\} \in E$ such that $i \in B$ or $j \in B$ define

\emph{Boundary condition:} $M^n_{j\to i}(x_i) := \mathbb{I}_{(x_j \ne x_i)}$ for all $n \ge 0$, if $j \in \partial B$,

\emph{Initialization:} $M^0_{j\to i}(x_i) := 0$ if $j \in B$,

\emph{Update:} If $j \in B$ and $n > 0$, then
\begin{eqnarray}
M^n_{j\rightarrow i}(x_i) :=  \min \limits_{x_j \in \{\pm 1\}}\left\{ \mathbb{I}_{(x_i\not=x_j)} + \sum\limits_{k\in\partial j\setminus i}M^{n-1}_{k\rightarrow j}(x_j)\right\} \nonumber
\end{eqnarray}
\end{definition}
Since the graph $B$ considered in this subsection is acyclic, this algorithm is referred to as Belief Propagation (BP),
or more specifically, as Min-Sum BP. After a number of iterations equal to the
length of the longest path in $B$, equation (\ref{unnormalized read out}) permits computation of $
O_i^* (\pm 1)$ for
each site in $B$. For cyclic graphs, such as the graph considered in this paper as defined in Section \ref{sec:setting of the paper}, the algorithm above can still be used and is referred to as \emph{Loopy} Belief Propagation.

\section{Difference messages and local solutions}\label{sec:local}

For the remainder of this paper, we are in the setting of Section \ref{sec:setting of the paper}. To avoid numerical overflow, it is standard practice to normalize the messages in LBP. In our case, we pass the \emph{difference messages} $m_{j\rightarrow i}^n := M_{j\rightarrow i}^n(-1) - M_{j\rightarrow i}^n(1)$ which satisfy the folowing recursion.

\begin{lemma}[\textbf{Difference Messages}]\label{difference message lemma}
	Definition \ref{unnormalized message} induces the following message-passing dynamics on the difference messages as follows: Given a boundary configuration $\bfx_{\partial B}$, for each edge $\{i,j\} \in E$ such that $i \in B$ or $j \in B$ we have
	
	\emph{Boundary condition}: $m^n_{j\to i} := x_j$ for all $n \ge 0$, if $j \in \partial B$,
	
	\emph{Initialization}: $m^0_{j\to i} := 0$ if $j \in B$,
	
	\emph{Update}: If $j \in B$ and $n > 0$, then 
	$$m^n_{j\to i} := \sign\left\{\sum_{k \in \partial j \setminus i}  m_{k\to j}^{n-1} \right\}$$ where $\sign(t) = t/|t|$ for $t \ne 0$ and $\sign(0) = 0$.
\end{lemma}

\iftoggle{full_version}{
\begin{proof}
	For the \emph{boundary condition}, if $j \in \partial B$, then $m_{j\rightarrow i}^n = M_{j\rightarrow i}^n(-1) - M_{j\rightarrow i}^n(1) = \mathbb{I}(x_j \ne -1)-\mathbb{I}(x_j \ne 1) = x_j$.
	
	There is nothing to check for the \emph{initialization}.
	
	For the \emph{update}, fix an $n > 0$ and an edge $\{i,j\} \in E$ such that $j \in B$. For $z \in \{\pm 1\}$, let $\Phi(z):= \sum_{k \in \partial j \setminus i} M_{k\to j}^{n-1}(z).$ By definition, we have
	\begin{align*}
	m^n_{j\to i} &= M_{j\to i}^{n}(1) - M_{j\to i}^{n}(-1)\\
	&=\min_{z \in \{\pm 1\}} \left\{\mathbb{I}(1 \ne z) + \Phi(z)\right\}  \\
	&\qquad -\min_{z \in \{\pm 1\}} \left\{\mathbb{I}(-1 \ne z) + \Phi(z)\right\}\\
	&=\min \left\{\Phi(1), 1 + \Phi(-1)\right\}  \\
	&\qquad -\min\left\{1 + \Phi(1), \Phi(-1)\right\}
	\end{align*}
	From the above and the fact that $\Phi_j$ maps into $\mathbb{Z}$, one gets
	\begin{align*} m^{n}_{j\to i} &=\begin{cases} 
	-1 & \Phi(1) > \Phi(-1) \\
	1 & \Phi(1) < \Phi(-1) \\
	0 & \Phi(1) = \Phi(-1)
	\end{cases}\\
	& = \sign\{\Phi(-1) - \Phi(1)\}
	\end{align*}
	By definitions, $\Phi(-1) - \Phi(1) = \sum_{k \in \partial j \setminus i} M_{k\to j}^{n-1}(-1) - M_{k\to j}^{n-1}(1) = \sum_{k \in \partial j \setminus i} m_{k\to j}^{n-1}.$ Thus, we are done. 
\end{proof}
}

By passing difference messages $\{m_{i\to j}^n\}$ rather than the original messages $\{M_{i\to j}^n\}$ we are unable to compute the quantities $O_i^*(1)$ and $O_i^*(-1)$. Nevertheless, we can use the difference messages to estimate the local solutions $o^*_i(\mathbf x_{\partial B})$ as defined by equation (\ref{definition: local solutions}). In Theorem \ref{main theorem}, we show that these estimates indeed converge to the truth. Below, we make the dependence of $o^*_i(\mathbf x_{\partial B})$ on the boundary condition $\bfx_{\partial B}$ implicit and simply write $o^*_i$. 


Before discussing the convergence of Min-Sum LBP, we first use the global solutions found in \cite{reyes2011cutset,reyes2014lossy} to directly compute the local solutions for sites $i\in B$. Using these local solutions, we will in Section \ref{sec:proof of main theorem} show that Min-Sum LBP converges to the correct local solutions for 1-run boundaries.

Let $\bfx_{\partial B}$ be a one-run boundary. A \emph{positive simple path} is a subset $P^+=\{i_1,\dots, i_k\}$ of $V$ such that 
\begin{enumerate}    
\item The subgraph $P^+$ is a path,
\item $i_1, i_k \in \partial B$ are the two endpoints of $R^+$,
\item for $1 < \ell < k$, either $i_\ell \in B$ or, if $i_\ell \in \partial B$, $x_{i_\ell} = 1$,
\item $k$ is minimal satisfying the above properties.
\end{enumerate}
A \emph{negative} simple path $P^-$ is defined in the same way by replacing $x_{i_\ell} = 1$ with $x_{i_\ell} = -1$ in item 3 above. Define the \emph{positive inner (resp.\ outer) path} $P^+_{I}$ (resp.\ $P^+_{O}$) be the positive simple path that minimizes (resp.\ maximizes) over all positive simple paths $P^+$ the number of nodes enclosed by $P^+\cup R^+$. Similarly, define $P^-_{I}$ and $P^-_{O}$. See examples in Figure \ref{inner and outer path 2}.


\begin{figure}[H]
\centering
\begin{subfigure}[b]{0.2\textwidth}
\includegraphics[width=1\textwidth]{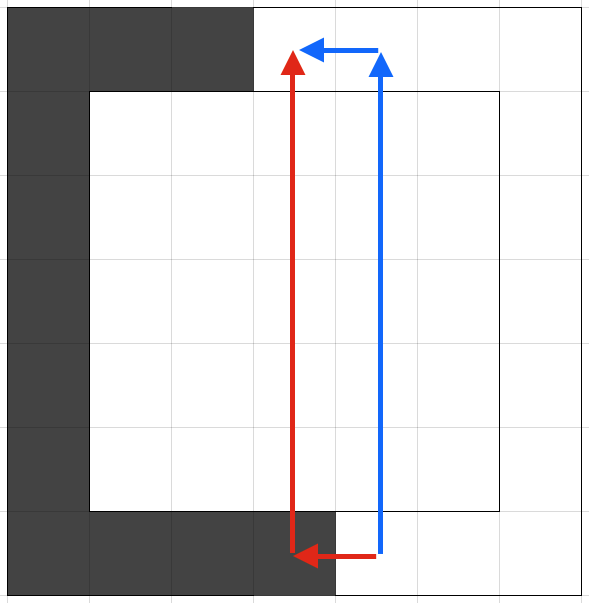}
\end{subfigure}
~
\begin{subfigure}[b]{0.2\textwidth}
\includegraphics[width=1\textwidth]{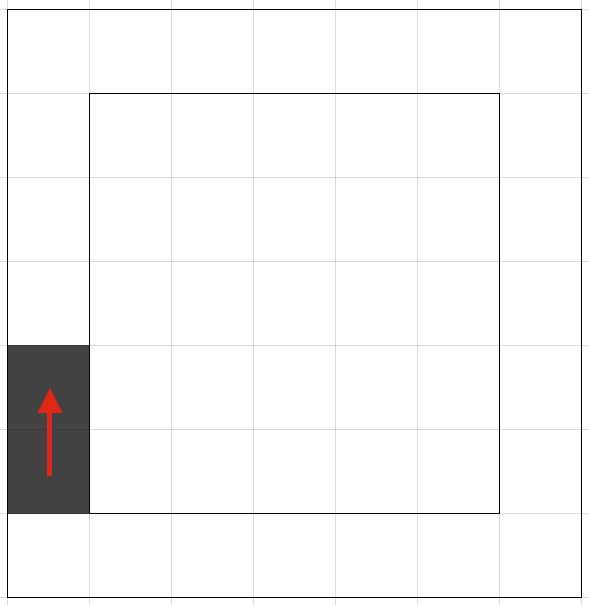}
\end{subfigure}
~
\begin{subfigure}[b]{0.2\textwidth}
\includegraphics[width=1\textwidth]{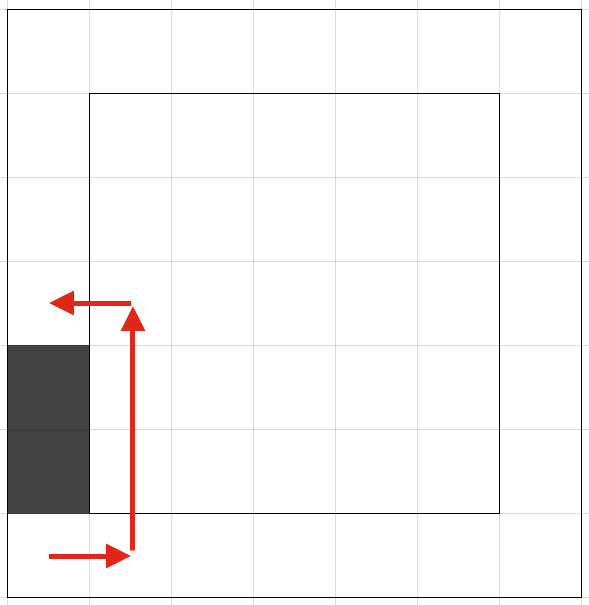}
\end{subfigure}
\caption{Left: Red path $= P_O^-$, blue path $= P_I^-$. Middle \& right: $P^+_O = P^+_I$.}\label{inner and outer path 2}
\end{figure}

Define $I^+$ to be the set $P^+_{I} \cup R^+$ together with the set of nodes enclosed by them. Similarly, define $O^+$, $I^-$ and $O^-$. We use the calligraphic font to denote intersection of these sets with $B$, e.g.\ $\mathcal I^+ = I^+ \cap B$, $\mathcal O^+ = O^+ \cap B$ and so on. We call $\mathcal I^+$ the \emph{positive inner region} $\mathcal O^+$ the \emph{positive outer region}, and so on. Define $\delta \mathcal{I}^+= \{ v \in \mathcal{I}^+ \mid v  \text{ touches } O^-\}$ and  $\delta \mathcal{I}^-= \{ v \in \mathcal{I}^+ \mid v  \text{ touches } O^+\}$. Note that $\delta \mathcal{I}^+$ is not the same as $\partial \mathcal{I}^+$, with which we are not concerned. The following proposition\iftoggle{full_version}{}{, whose proof is in the full version, } gives the local solutions over regions that we have defined.

\begin{proposition}\label{correctness oracle}
	For a one-run boundary $\bfx_{\partial B}$,
		$$
	o^*_i(\bfx_{\partial B})  = \begin{cases}
	\pm 4 &: i\in \mathcal{I}^\pm \setminus \delta \mathcal{I}^\pm  \\
	\pm 2 &:i\in \delta \mathcal{I}^\pm  \\
	0 &:  i\in \mathcal{O}^+ \cap \mathcal{O}^-
	\end{cases}
	$$
are the local solutions for sites $i\in B$.
\end{proposition}

\iftoggle{full_version}{
\begin{proof}

    Let $i\in \mathcal{O}^+ \cap \mathcal{O}^-$. Thus $i\in \mathcal{O}^+$, which means that there is a MAP configuration, namely the one generated by the black outer path, in which site $i$ is black. Likewise, $i\in \mathcal{O}^-$, which means that there is a MAP configuration in which site $i$ is white. Therefore, $O^*_i(-1) = O^*_i(1)$, and hence $o^*_i = 0$.

    Now let $i\in \mathcal{I}^+ \setminus \delta \mathcal{I}^+$. Since $i\in \mathcal{I}^+$, site $i$ is black in all MAP configurations, therefore $O^*_i(1) < O^*_i(-1)$. Moreover, since $i\notin \delta \mathcal{I}^+$, in all MAP configurations, every neighbor $j\in\partial i$ is also black. Therefore, in any MAP configuration, if we flip site $i$ to white, then, this will increase the number of odd bonds by 4. If we also flip a neighbor $j\in\partial i$ to white, this will further increase the number of odd bonds by at least 2. Therefore $O^*_i(-1) = O^*_i(1) + 4$.

    Let $i\in \delta \mathcal{I}^+$. Since $i\in \mathcal{I}^+$, site $i$ is black in all MAP configurations, therefore $O^*_i(1) < O^*_i(-1)$. By definition of $\delta \mathcal{I}^+$, there is a MAP configuration $x_B$ in which a neighbor $j\in\partial i$ is white. We claim that there can only be one such neighbor $j$. This is because if there were two such neighbors, then we could flip $i$ to white and keep the number of odd bonds the same, which would imply that there is a MAP configuration in which $i$ is white. If there were more than two such neighbors $j\in\partial i$, then flipping $i$ to white would strictly decrease the number of odd bonds, which contradicts $O^*_i(1) < O^*_i(-1)$. Thus, $O^*_i(-1) = O^*_i(1) + 2$.

    The remaining two cases follow from arguments analogous to those in the previous two paragraphs. This completes the proof. \end{proof}
}

For each $i \in B$, define $\hat o^n_i   :=  \sum_{j\in\partial i} m^{n}_{j\rightarrow i}$ to be the \emph{estimates} of the local solution $o^*_i$. Applying the usual arguments for correct convergence of BP on trees (see \cite{pearl2014probabilistic} for instance) gives the following proposition

\begin{proposition}\label{tree result} Let $G=(V,E)$ be a tree and $B\subseteq V$ be a subtree, and $\bfx_{\partial B}$ be any boundary configuration. Then the difference messages $m^n_{j\to i}$ converge and the estimates $\hat o^n_i$
	converge to $o^*_i$ in number of steps equal to the diameter of the tree plus 1.
\end{proposition}

We now present our main result.

\begin{theorem}\label{main theorem}
Let $G= (V,E)$ and $B\subseteq V$ be as defined in Section \ref{sec:setting of the paper}, and $\bfx_{\partial B}$ be a one-run boundary configuration. Then for every edge $\{j,i\}$ in $G$ with $j \in B$, the difference messages $m^n_{j\to i}$ converge and the estimates $\hat o^n_i$
converge to $o^*_i$ in number of steps equal to $2N$.\end{theorem}

We note that the diameter of $B$ is $2N-2$ rather than $2N$. Thus, Theorem \ref{main theorem} parallels Proposition \ref{tree result}.

\section{Forward and backward convergence} \label{sec:forbackstab}

In this section, we present the two technical lemmas, the Forward and Backward Convergence Lemmas, that allow us to prove the convergence of the difference messages in the setting of Section \ref{sec:setting of the paper}. The main idea is that convergence of messages on certain rectangular subsets of $B$ takes place in two phases, forward from the corners of the graph, and then {backward towards the boundary}.

Recall that we are in the setting of Section \ref{sec:setting of the paper}. We identify $V$ with the point set $\{0,1,\dots, N\}^2$ in the usual way, i.e., $(0,0)$ is the bottom-left corner, $(0,N+1)$ is the top-left corner, and so on. We represent nodes in $V$ in two ways. The first way uses $i,j,k,l, m$ (with possible subscripts) to represent a vertex in a coordinate-free way, e.g., $i_1,i_2 \in V$. The second way uses $(a,b),(\alpha,\beta)$ (again, with possible subscripts) to represent a vertex in coordinates, e.g.\ $(a,b) \in V$ where $a, b \in \{0,1,\dots,N+1\}$. At times, we will say ``let $i=(a,b) \in V$ be a vertex" to simultaneously refer to both representations.

Define $\vec{G} = (V,\vec{E})$ to be the directed graph such that the vertex set of $\vec{G}$ and $G = (V,E)$ are the same, and for each undirected edge $\{i,j\} \in E$, there are exactly two directed edges $i \to j$ and $j\to i$ in $\vec{E}$. Let 
$\{\mathcal{N}, \mathcal{S}, \mathcal{E}, \mathcal{W}\}$ be the set of \emph{directions} north, south, east, and west, respectively. Define corresponding direction vectors $v(\mathcal N) = (0,1)$, $v(\mathcal S) = (0,-1)$, $v(\mathcal E) = (1,0)$ and $v(\mathcal W) = (-1,0)$. For each $(a,b) \in V$ and direction $\mathcal{D}$, if $(a,b) + v(\mathcal D) \in V$, where $+$ is just the usual vector summation, then $(a,b) + v(\mathcal D)$ is said to be the \emph{$\mathcal D$ neighbor} of $(a,b)$. For example, $(a+1,b)$ is the eastern neighbor of $(a,b)$. Given a node $(a,b)\in V$ and a direction $\mathcal{D}$ such that $(a,b) + v(\mathcal D) \in V$, define $ \mathcal D(a,b) \in \vec{E}$ as $ \mathcal D(a,b) := (a,b) \to (a,b) + v(\mathcal D).$ 
For example, $\mathcal E(a,b) = (a,b) \to (a+1,b)$. 

For a message $m^n_{j\to i}$ as defined in Lemma \ref{difference message lemma}, we often drop the superscript and simply write $m_{j\to i}$ when the time index is not of concern. With this notation, $m_{\mathcal E(a,b)}$ denotes the message from $(a,b)$ to its eastern neighbors and so on. Now, given a subset $S \subseteq B$ and a direction $\mathcal{D}$, a {message} \emph{received by $S$ from direction $\mathcal{D}$} is a {message} of the form $m_{\mathcal{D}(a,b)}$ for some $(a,b) \not\in S$ such that $(a,b) + v(\mathcal{D}) \in S$. A message \emph{sent from $S$ in the direction $\mathcal{D}$} is a {message} of the form $m_{\mathcal{D}(a,b)}$ for some $(a,b) \in S$. Note that according to these defintion, a message $m_{\mathcal{D}(a,b)}$ sent from $S$ is not considered to be received by $S$ even if $(a,b) + v(\mathcal D)$ is in $S$. See Figure \ref{figure: FC and BC} which illustrates the messages sent from and received by $S$ where $S$ is the set of blue nodes.

\iftoggle{full_version}{}{
\begin{figure}[H]
\centering
\begin{subfigure}[b]{0.15\textwidth}
\includegraphics[width=1\textwidth]{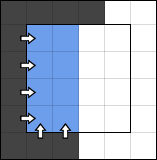}
\end{subfigure}
\begin{subfigure}[b]{0.15\textwidth}
\includegraphics[width=1\textwidth]{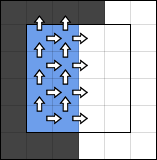}
\end{subfigure}
\begin{subfigure}[b]{0.15\textwidth}
\includegraphics[width=1\textwidth]{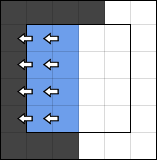}
\end{subfigure}
\caption{Left: messages received by $R_{(1,1)}^{(2,4)}$ from the direction $\mathcal{N}$ and $\mathcal{E}$. Middle: convergent messages due to Lemma \ref{forward conv}. Right: convergent messages due to Lemma \ref{lemma:backstab}.}\label{figure: FC and BC}
\end{figure}
}

The sets $S$ for which we are interested in messages received by and sent from are rectangular subsets of $B$ defined by the boundary runs. Let $i_1 = (\alpha_1,\beta_1)$ and $i_2 = (\alpha_2,\beta_2) \in B$. Define the \emph{rectangle $R_{i_1}^{i_2}$ with corners at $i_1$ and $i_2$} by
\begin{align*} R_{i_1}^{i_2} := \{(a,b) \in B : &\min(\alpha_1,\alpha_2) \le a \le \max(\alpha_1,\alpha_2),\\
&\min(\beta_1,\beta_2) \le b \le \max(\beta_1,\beta_2)\}.
\end{align*}

\iftoggle{full_version}{
For $D \in \mathbb{N}$, define the \emph{cut-rectangle} $R_{i_1}^{i_2}(D)$ of nodes of distance $D-1$ from the corner $i_1$ as
$$R_{i_1}^{i_2}(D):= \{(a,b) \in R_{i_1}^{i_2} :
|a- \alpha_1|+|b-\beta_1| \le D-1\}.
$$

Define the \emph{L-shaped region} $L_{i_1}^{i_2}$ with corner at $i_1$ to be
\begin{align*} L_{i_1}^{i_2} := & \{(\alpha_1,b) \in B :\min(\beta_1,\beta_2) \le b \le \max(\beta_1,\beta_2)\}\\
&\cup \{(a, \beta_1) \in B :\min(\alpha_1,\alpha_2) \le a \le \max(\alpha_1,\alpha_2)\}.
\end{align*}

See Figure \ref{rectangle examples} for examples.

\begin{figure}[H]
\centering
    \includegraphics[width=0.5\textwidth]{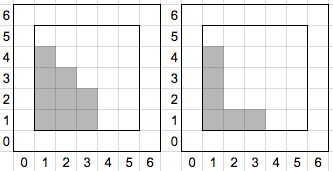}
    \caption{Left: $R_{(1,1)}^{(3,4)}(4)$. Right: $L_{(1,1)}^{(3,4)}$.}\label{rectangle examples}
\end{figure}
}

We say the \emph{unordered} pair of directions $(\mathcal{D}_1,\mathcal{D}_2)$ is \emph{adjacent} if $v(\mathcal D_1)$ and $v(\mathcal D_2)$ are orthogonal. The set of adjacent pairs of directions is \iftoggle{full_version}{$$\mathcal A := \{ (\mathcal{E},\mathcal{N}), (\mathcal{E},\mathcal{S}),
(\mathcal{W},\mathcal{N}), (\mathcal{W},\mathcal{S})\}.$$}{$\mathcal A := \{ (\mathcal{E},\mathcal{N}), (\mathcal{E},\mathcal{S}),
(\mathcal{W},\mathcal{N}), (\mathcal{W},\mathcal{S})\}.$} For two elements $i_1,i_2 \in B$, we use the notation $i_1 \triangleright i_2$ to denote an \emph{ordered} pair of vertices which are not necessarily neighbors.
\begin{definition}[\textbf{Compatible tuples}]Let $i_1=(\alpha_1,\beta_1),\, i_2=(\alpha_2,\beta_2) \in B$.
We say that the tuple $(i_1 \triangleright i_2, \mathcal D_1, \mathcal D_2)$ is \emph{compatible} if $(\mathcal D_1, \mathcal D_2) \in \mathcal A$ and there exists non-negative coefficients $c_1,c_2$ such that
$c_1 v(\mathcal D_1) + c_2 v(\mathcal D_2) = (\alpha_2-\alpha_1, \beta_2-\beta_1)$.
\end{definition}
The ordering of $i_1 \triangleright i_2$ is crucial because if $i_1 \ne i_2$, then $(i_1 \triangleright i_2, \mathcal D_1, \mathcal D_2)$ is compatible implies  $(i_2 \triangleright i_1, \mathcal D_1, \mathcal D_2)$ is \emph{not} compatible.  However, the ordering of $\mathcal{D}_1$ and $\mathcal{D}_2$ is irrelevant.



 \begin{definition}[\textbf{Convergence of messages}]Let $\sigma =\pm 1$ and $N \in \mathbb{N}_0$. For a given directed edge $j\to i \in \vec{E}$, we say that the messages $m^n_{j \to i}$  \emph{converges in $N$ iterations to} $\sigma$ if $m^n_{j\to i} = \sigma$ for all $n \ge N$.\end{definition}
 

The definition below is the salient feature of the one-run boundary configuration underlying the proof of our main result

\begin{definition}[\textbf{Forward Convergence (FC) Property}]\label{definition: FCP}
The compatible tuple $(i_1\triangleright i_2, \mathcal D_1, \mathcal D_2)$ is \emph{forward convergent to $\sigma=\pm 1$ at time $n_0 \in \mathbb{N}_0$}, abbreviated as $\mathrm{FC}(\sigma, n_0)$, if each message received by $R_{i_1}^{i_2}$ from directions $\mathcal{D}_1$ and $\mathcal{D}_2$ converge in  $n_0$ iterations to $\sigma$.
\end{definition}


We will refer to the following as the \emph{FC Lemma}.

\begin{lemma}[\textbf{Forward Convergence}]\label{forward conv}Let $\sigma =\pm 1$ and $n_0 \in \mathbb{N}_0$. Suppose that $(i_1\triangleright i_2, \mathcal D_1, \mathcal D_2)$ is $\mathrm{FC}(\sigma,n_0)$. Then all messages sent from $R_{i_1}^{i_2}$ in the directions $\mathcal{D}_1$ and $\mathcal{D}_2$ converge in $(n_0+2N-1)$ iterations to $\sigma$.
\end{lemma}

\iftoggle{full_version}{
\begin{proof}
Letting $D = 2N-1$ in Lemma \ref{forward conv lemma} (proven next), then $R_{i_1}^{i_2} = R_{i_1}^{i_2}(D)$. Thus, we have the desired result.
\end{proof}
}

\iftoggle{full_version}{
Lemma \ref{forward conv} is essentially a corollary of the lemma below.

\begin{lemma}\label{forward conv lemma}
Let $\sigma =\pm 1$ and $n_0 \in \mathbb{N}_0$. Suppose $(i_1\triangleright i_2,\mathcal D_1, \mathcal D_2)$ is $\mathrm{FC}(\sigma,n_0)$. Then for all $D \in \mathbb{N}$, all edges sent from $R_{i_1}^{i_2}(D)$ in the directions $\mathcal{D}_1$ and $\mathcal{D}_2$ converge in $(n_0+D)$ iteration to $\sigma$.
\end{lemma}
\begin{proof}
For $\ell=1,2$, let $i_\ell = (\alpha_\ell,\beta_\ell)$ be the coordinate representation of $i_\ell$. Since $V$ is embedded in $\mathbb{R}^2$, rotations and reflections on $\mathbb{R}^2$ that preserve $V$ induce graph automorphisms on $G= (V,E)$ that respects the adjacency of directions. Thus, because the message-passing dynamics defined in Lemma \ref{difference message lemma} is isotropic, we can assume $\alpha_1 \le \alpha_2$, $\beta_1 \le \beta_2$ and $\mathcal{D}_1 = \mathcal E$ and $\mathcal D_2 = \mathcal N$  after applying appropriate rotations and reflections. Relabeling $V$ using the rule $(a,b) \mapsto (a-\alpha_1+1, b - \beta_1+1)$, we may further assume $\alpha_1 = \beta_1 = 1$. Relabeling time, we may assume $n_0 = 0$. The dynamics defined in Lemma \ref{difference message lemma} is invariant under multiplication by $-1$, i.e., we may replace every instance of $m^n_{j\to i}$ by $-m^n_{j\to i}$. Hence, we may assume $\sigma = 1$. Let $(\alpha, \beta) = (\alpha_2,\beta_2)$.

Given the reductions in the preceding paragraph, we now have only to prove that for each $(a,b) \in R_{(1,1)}^{(\alpha,\beta)}(D)$ such that $a \le \alpha$ and $b \le \beta$, the messages $m^n_{\mathcal E(a,b)}$ and $m^n_{\mathcal N(a,b)}$ converge in $D$ iterations to $1$. In other words, for all $n \ge D$
\begin{align}
    \label{FC: east}
    m^n_{\mathcal E(a,b)} := \sign \{m^{n-1}_{\mathcal E(a-1,b)} +m^{n-1}_{\mathcal N(a,b-1)} + m^{n-1}_{\mathcal S(a,b+1)} \} =1\\
    \label{FC: north}
    m^n_{\mathcal N(a,b)} := \sign \{m^{n-1}_{\mathcal E(a-1,b)} + m^{n-1}_{\mathcal N(a,b-1)}+m^{n-1}_{\mathcal W(a+1,b)}\} =1
\end{align}
In order to prove (\ref{FC: east}) and (\ref{FC: north}), it suffices to show 

\begin{equation}\label{forward convergence lemma: to be checked}m^{n-1}_{\mathcal E(a-1,b)}=m^{n-1}_{\mathcal N(a,b-1)} = 1
\end{equation} 
because in this case the value of the third term inside the $\sign$ in (\ref{FC: east}) and (\ref{FC: north}) above cannot influence the message.

The boundary conditions part of the definition of property $\mathrm{FC}(1,0)$ translates to the fact that 
\begin{equation}\label{forward conv lemma: boundary condition}m^n_{\mathcal E(0,b)} = m^n_{\mathcal N(a,0)} = 1, \, \forall n\ge 0\end{equation}
for all $1 \le a \le \alpha$ and $1 \le b \le \beta$.

We proceed by induction on $D$. For the base case, $D = 1$ and $R_{(1,1)}^{(\alpha,\beta)}(D) = \{(1,1)\}$. Let $n \ge 1$. Observe that (\ref{forward convergence lemma: to be checked}) follows from (\ref{forward conv lemma: boundary condition}). Thus, (\ref{FC: east}) and (\ref{FC: north}) are proven for $D = 1$ and $n \ge D$.

Now, let $D > 1$ and suppose the conclusion of Lemma \ref{forward conv lemma} holds for $D -1$. Let $(a,b) \in R_{(1,1)}^{(\alpha,\beta)}(D)$ and $n \ge D$. If $(a,b) \in R_{(1,1)}^{(\alpha,\beta)}(D-1)$, then by the induction hypothesis, we're done. Thus, below, we assume $(a,b) \in R_{(1,1)}^{(\alpha,\beta)}(D) \setminus R_{(1,1)}^{(\alpha,\beta)}(D-1)$, that is, $|a-1|+|b-1| = D-1$. Our goal as before is to show (\ref{forward convergence lemma: to be checked}).

First, consider the case that $a > 1$ and $b > 1$. Then $(a-1,b), (a,b-1) \in R_{(1,1)}^{(\alpha,\beta)}(D-1)$ and so by the induction hypothesis $m^{n-1}_{\mathcal E(a-1,b)} =m^{n-1}_{\mathcal N(a,b-1)} =1$ and so (\ref{forward convergence lemma: to be checked}) holds.

Next, suppose $a = 1$. Then $|b-1| = D-1>0$ implies that $b > 1$. From the boundary condition (\ref{forward conv lemma: boundary condition}), we have $m^{n-1}_{\mathcal E(a-1,b)}  = 1$. On the other hand, $(a,b-1) \in R_{(1,1)}^{(\alpha,\beta)}(D-1)$ and so $m^{n-1}_{\mathcal N(a,b-1)} = 1$, which proves (\ref{forward convergence lemma: to be checked}).

The case when $b = 1$ is analogous to the above argument. This proves the induction step and the lemma.
\end{proof}

}

%

The following definition builds on the notion of the FC property defined previously.

%

\begin{definition}[\textbf{Backward Convergence (BC) Property}]
Two compatible tuples $(i_1\triangleright i_2, \mathcal D_1, \mathcal D_2)$ and $(i_3\triangleright i_4,\mathcal D_3, \mathcal D_4)$ are said to be \emph{backward convergent to $\sigma \in \pm 1$ at time $n_0 \in \mathbf N_0$}, abbreviated as $\mathrm{BC}(\sigma, n_0)$, if both tuples are $\mathrm{FC}(\sigma, n_0)$ and there is exactly one direction in common, i.e., $|\{\mathcal D_1, \mathcal D_2\} \cap \{\mathcal D_3, \mathcal D_4\}| = 1$. The unique element $\cal{D}$ \emph{not} in $\{\mathcal D_1, \mathcal D_2\} \cup \{\mathcal D_3, \mathcal D_4\}$ is called the \emph{backward convergence (BC) direction}.
\end{definition}

\begin{lemma}[\textbf{Backward Convergence}]\label{lemma:backstab}
Let $\sigma =\pm 1$ and $n_0 \in \mathbb{N}_0$. Suppose $(i_1\triangleright i_2, \mathcal D_1, \mathcal D_2)$  and $(i_3\triangleright i_4,\mathcal D_3, \mathcal D_4)$ are $\mathrm{BC}(\sigma, n_0)$ and let $\mathcal{D}$ be the BC direction. Then all messages sent from $R_{i_1}^{i_2} \cap R_{i_3}^{i_4}$ in the direction $\mathcal{D}$ converge in $(n_0+2N)$ iterations to $\sigma$.
\end{lemma}

\iftoggle{full_version}{
\begin{proof}
Let $i_\ell = (\alpha_\ell,\beta_\ell)$ for $\ell = 1,2,3,4$. Without the loss of generality, we consider the case when $\mathcal D_1 = \mathcal{N}, \mathcal D_2 = \mathcal D_4= \mathcal{E}$ and  $\mathcal D_3 = \mathcal{S}$. Hence, $\mathcal{D} = \mathcal{W}$. Furthermore, we assume $\sigma = 1$ and $n_0 = 0$. Now, let $(a,b) \in R_{(\alpha_1,\beta_1)}^{(\alpha_2,\beta_2)} \cap R_{(\alpha_3,\beta_3)}^{(\alpha_4,\beta_4)}$. Let $n \ge 2N$. Our goal is to show 
\begin{equation}\label{backward convergence lemma: main check}m^n_{\mathcal W(a,b)} := \sign \{m^{n-1}_{\mathcal N(a,b-1)} + m^{n-1}_{\mathcal S(a,b+1)} +m^{n-1}_{\mathcal W(a+1,b)}\} =1\\
\end{equation}
We first show $m^{n-1}_{\mathcal N(a,b-1)} = 1$. Now, if $(a,b-1) \in R_{(\alpha_1,\beta_1)}^{(\alpha_2,\beta_2)}$, then $m^{n-1}_{\mathcal N(a,b-1)} = 1$ by FC Lemma \ref{forward conv}. On the other hand, if $(a,b-1) \not\in R_{(\alpha_1,\beta_1)}^{(\alpha_2,\beta_2)}$, then $\min(\beta_1,\beta_2) = b$. Since $(\mathcal{E}, \mathcal{N})$ is compatible with $((\alpha_1,\beta_1),(\alpha_2,\beta_2))$ by assumption, we have $\beta_1 =\min(\beta_1,\beta_2)$. This shows $(a,b) =(a,\beta_1) \in L_{(\alpha_1,\beta_1)}^{(\alpha_2,\beta_2)}$. Next, since $(a,b-1) + v(\mathcal{N}) = (a,b) \in R_{(\alpha_1,\beta_1)}^{(\alpha_2,\beta_2)}$, $m^n_{\mathcal{N}(a,b-1)}$ is a message received by $L_{(\alpha_1,\beta_1)}^{(\alpha_2,\beta_2)}$ from the direction $\mathcal{N}$. Thus, the defining property of $((\alpha_1,\beta_1)\triangleright(\alpha_2,\beta_2), \mathcal N, \mathcal E)$ being $\mathrm{FC}(1, 0)$ to obtain  $m^{n-1}_{\mathcal N(a,b-1)} = 1$. 

An analogous argument shows $m^{n-1}_{\mathcal S(a,b+1)} = 1$. This proves (\ref{backward convergence lemma: main check}).
\end{proof}
}

See Figure \ref{figure: FC and BC} for an illustration of Lemma \ref{forward conv} and \ref{lemma:backstab}.

\iftoggle{full_version}{
\begin{figure}[H]
\centering
\begin{subfigure}[b]{0.2\textwidth}
\includegraphics[width=1\textwidth]{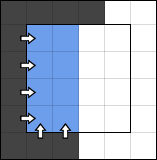}
\end{subfigure}
~
\begin{subfigure}[b]{0.2\textwidth}
\includegraphics[width=1\textwidth]{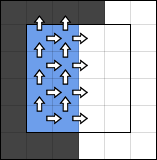}
\end{subfigure}
~
\begin{subfigure}[b]{0.2\textwidth}
\includegraphics[width=1\textwidth]{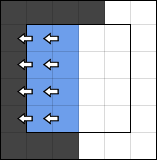}
\end{subfigure}
\caption{Left: messages received by $R_{(1,1)}^{(2,4)}$ from the direction $\mathcal{N}$ and $\mathcal{E}$. Middle: convergent messages due to Lemma \ref{forward conv}. Right: convergent messages due to Lemma \ref{lemma:backstab}.}\label{figure: FC and BC}
\end{figure}
}



\section{Proof of Theorem \ref{main theorem}}\label{sec:proof of main theorem}

Let $\bfx_{\partial B}$ be a given one-run boundary. Below, we assume the algorithm has run for $2N$ iterations, i.e., $n\ge 2N$, so that we can use the FC and BC Lemma. The goal is to show that $\hat o^n_{(a,b)} = o^*_{(a,b)}$ for any $(a,b) \in B$ where $o^*_{(a,b)}$ is as in Proposition \ref{correctness oracle} and 
$$\hat o^n_{(a,b)} =  m^n_{\mathcal{E}(a-1,b)} 
+ m^n_{\mathcal{N}(a,b-1)}
+ m^n_{\mathcal{W}(a+1,b)}
+ m^n_{\mathcal{S}(a,b+1)}.
$$ Without the loss of generality, we assume  $\bfx_{\partial B}$ satisfies:

    
    \emph{1. Positive run is contracted: } Suppose $i,j,k \in \partial B$ are such that $j$ has degree $2$, $i$ and $k$ are the two neighbors of $j$, and $x_i \ne x_k$. In such cases, we always assume $x_j = -1$. This is because $j$ does not touch any nodes in $B$, so the value of $x_j$ does not affect the message-updates. 

     \emph{2. Positive run is smaller} i.e., $|R^+| \le |R^-|$.


A node in $B$ with two neighbors in $\partial B$ is called a \emph{corner}. For easier visualization, the four corners are given names: $sw = (1,1)$, $se = (N,1)$, $ne = (N,N)$ and $nw = (1,N)$.
Define $C$ to be the set of corners with two positive boundary neighbors, i.e., $$C = \{ i \in \{sw,se,ne,nw\} :  |\{j \in \partial B \cap \partial i \mid x_j = +1 \}| = 2\}$$
Notice that $|C| \le 2$ because otherwise $|R^+| > |R^-|$. Thus, $|C| \in \{0,1,2\}$ and the proof is correspondingly divided into the three cases. \iftoggle{full_version}{}{Due to the lack of space, we will only show the proof for $|C| = 1$.  For the full proof, we refer the reader to \cite{}[{\color{red}full arxiv paper}].}

Below, for brevity, we will write subsets of $B$ using probabilist notation, i.e.,  for a logical statement $S$, let $\{S\}$ be a shorthand for $\{(a,b) \in B : S\}$. For example, $\{(a,b) \in B : a = 1\}$ is simply written as $\{a=1\}$.

\iftoggle{full_version}{
Case $|C| = 0$. Let $(\alpha_1,\beta_1), (\alpha_2,\beta_2)$ be the two end points of $R^+$ such that $\beta_1 \le \beta_2$.
 Without loss of generality, we assume all the positive boundary conditions are restricted to the left side. More precisely, if $(a,b) \in \partial B$ and $x_{(a,b)} = 1$, then $a = 0$. See Figure \ref{fig:Case 0} for an example.

\begin{figure}[H]
\centering
\begin{subfigure}[t]{0.3\textwidth}
\includegraphics[width= 1 \textwidth]{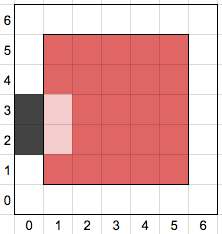}
\end{subfigure}
~
\begin{subfigure}[t]{0.15\textwidth}
\includegraphics[width=1\textwidth]{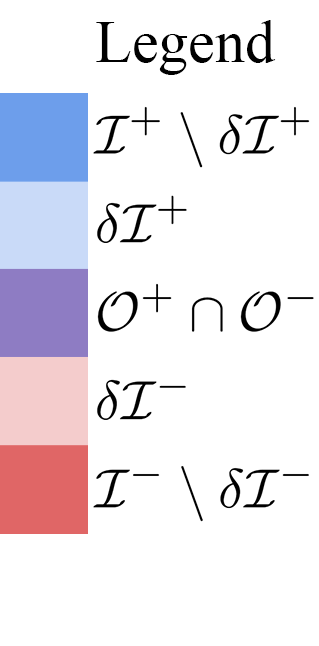}
\end{subfigure}
\caption{Example of case $C = \emptyset$ where $(\alpha_1,\beta_1) = (0,2), (\alpha_2,\beta_2)=(0,3)$.}\label{fig:Case 0}
\end{figure}

From the definitions, it is easy to see that
\begin{align*}
	&\mathcal{I}^+ \setminus \delta \mathcal{I}^+ =\delta \mathcal{I}^+ = \mathcal{O}^+ \cap \mathcal{O}^-= \emptyset\\
	&\delta \mathcal{I}^-  = \{a = 1,\,\beta_1 \le b \le  \beta_2\}\\
	&\mathcal{I}^- \setminus \delta \mathcal{I}^- = \{1< a\}\cup \{b > \beta_2\} \cup \{b < \beta_1\}
\end{align*}

Observe that $(ne\triangleright sw, \mathcal S, \mathcal W)$ and $(se\triangleright nw,\mathcal N, \mathcal W)$ are $BC(-1,0)$ with backward convergence direction $\mathcal{E}$. We will refer to this observation as $\dagger$. Thus, for all $(a,b) \in B$, we have $m^n_{\mathcal{N}(a,b-1)}
= m^n_{\mathcal{W}(a+1,b)}
= m^n_{\mathcal{S}(a,b+1)}
= -1$. We only have to analyze $m^n_{\mathcal{E}(a-1,b)}$.

Let $(a,b) \in \mathcal{I}^- \setminus \delta \mathcal{I}^-$. If $1<a$, then by $\dagger$, $m^n_{\mathcal{E}(a-1,b)} = -1$ and so $\hat o^n_{(a,b)} = -4$. If $a = 1$, then $m^n_{\mathcal{E}(a-1,b)} = -1$ by the boundary condition. So in both cases, $\hat o^n_{(a,b)} = -4$.

Next, consider $(a,b) \in \delta \mathcal{I}^-$, then $m^n_{\mathcal{E}(a-1,b)} = 1$ by the boundary condition. Hence, $\hat o^n_{(a,b)} = -2$.
}



\begin{figure}[H]
\centering
\begin{subfigure}[t]{0.3\textwidth}
\includegraphics[width=1\textwidth]{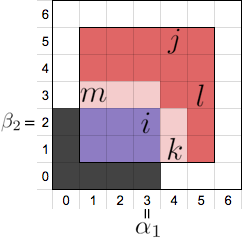}
\end{subfigure}
~
\begin{subfigure}[t]{0.15\textwidth}
\includegraphics[width=1\textwidth]{legend}
\end{subfigure}
    \caption{$|C| = 1$ example, $(\alpha_1,\beta_1) = (3,0), (\alpha_2,\beta_2)=(0,2)$.}
    \label{fig:corner1}
\end{figure}
\iftoggle{full_version}{
Case $|C| = 1$: 
}

Without loss of generality, let $C = \{sw\}$ and let $(\alpha_1,\beta_1)$ and $(\alpha_2,\beta_2)$ be the two endpoints of $R^+$ such that $\beta_1 \le \beta_2$.

One checks that $\mathcal{I}^+ \setminus \delta \mathcal{I}^+ =\delta \mathcal{I}^+ = \emptyset$ and
\begin{align*}
& \mathcal{O}^+ \cap \mathcal{O}^-= \{a \le \alpha_1, \, b \le \beta_2\}\\
&\delta \mathcal{I}^-  = \{a = \alpha_1+1, \, b \le  \beta_2\}\cup \{a \le \alpha_1, b=\beta_2 +1\}\\
&\mathcal{I}^- \setminus  \delta \mathcal{I}^-= \{a > \alpha_1 +1\}\cup \{b > \beta_2+1\} \\
&\qquad\qquad \qquad  \cup\{a=\alpha_1+1, b=\beta_2+1\}
\end{align*}
Let $i=(\alpha_1,\beta_2)$, $j=(\alpha_1+1,N)$, $k=(\alpha_1+1,1)$, $l=(N,\beta_2+1)$, and $m = (1,\beta_2+1)$. See Figure \ref{fig:corner1}. We have the following observations
\begin{enumerate}[label=\Roman*]
\item\label{case 2: lower left} $(sw\triangleright i, \mathcal E, \mathcal N )$ is $\mathrm{FC}(1,0)$,
\item\label{case 2: upper right} $(ne\triangleright sw, \mathcal W, \mathcal S )$ is $\mathrm{FC}(-1,0)$,
\item\label{case 2: right} $(se\triangleright j, \mathcal N, \mathcal W)$  and $(ne\triangleright k, \mathcal S, \mathcal W)$ are $\mathrm{BC}(-1, 0)$,
\item\label{case 2: top} $(nw\triangleright l, \mathcal E, \mathcal S)$  and $(ne\triangleright m, \mathcal S, \mathcal W)$ are $\mathrm{BC}(-1, 0)$.
\end{enumerate}


If $(a,b) \in \mathcal{O}^+ \cap \mathcal{O}^-$, then $m^n_{\mathcal{E}(a-1,b)} = m^n_{\mathcal{N}(a,b-1)} = 1$ by \ref{case 2: lower left} and $m^n_{\mathcal{W}(a+1,b)} = m^n_{\mathcal{S}(a,b+1)} = -1$ by \ref{case 2: upper right}, we have $\hat o^*_{(a,b)} = 0$.

Let $(a,b) \in \delta \mathcal{I}^-$. If $(a,b) \in  \{a = \alpha_1+1, \, b \le  \beta_2\}$, then $m^n_{\mathcal{E}(a-1,b)} =m^n_{\mathcal{N}(a,b-1)}=m^n_{\mathcal{S}(a,b+1)} = -1$  by \ref{case 2: right}, and $m^n_{\mathcal{W}(a+1,b)}=1$ by \ref{case 2: lower left}. On the other hand, if $(a,b) \in \{a \le \alpha_1, b=\beta_2 +1\}$, then 
$m^n_{\mathcal{E}(a-1,b)}
= m^n_{\mathcal{N}(a,b-1)}
= m^n_{\mathcal{W}(a+1,b)}
= -1
$ by \ref{case 2: top}, and $m^n_{\mathcal{W}(a+1,b)}=1$ by \ref{case 2: lower left}. Hence, in both cases, $\hat o^n_{(a,b)} = -2$.

Finally, let $(a,b) \in \mathcal{I}^- \setminus  \delta \mathcal{I}^-$. If $(a,b) = (\alpha_1+1,\beta_2+1)$, then $
 m^n_{\mathcal{E}(a-1,b)} = m^n_{\mathcal{S}(a,b+1)}=-1$ by \ref{case 2: top} and 
 $
 m^n_{\mathcal{N}(a,b-1)}=
 m^n_{\mathcal{W}(a+1,b)}=-1
 $
 by \ref{case 2: right}. Hence, $\hat o^n_{(a,b)} = -4$. If $b > \beta_2+1$, then every messages received by $(a,b)$ is equal to $-1$ by \ref{case 2: top}, so $\hat o^n_{(a,b)} = -4$. Likewise, if $a > \alpha_1 +1$, then every messages received by $(a,b)$ is equal to $-1$ by \ref{case 2: right}, so $\hat o^n_{(a,b)} = -4$.



\iftoggle{full_version}{
Case $|C| = 2$: without loss of generality, let $C= \{sw,nw\}$ and let $C = \{sw\}$ and let $(\alpha_1,\beta_1)$ and $(\alpha_2,\beta_2)$ be the two endpoints of $R^+$ such that $\alpha_1 \le \alpha_2$.

\begin{figure}[H]
\centering
\begin{subfigure}[t]{0.3\textwidth}
\includegraphics[width= 1 \textwidth]{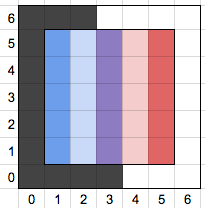}
\end{subfigure}
~
\begin{subfigure}[t]{0.15\textwidth}
\includegraphics[width=1\textwidth]{legend}
\end{subfigure}
        \caption{$|C| = 2$ example where $(\alpha_1,\beta_1) = (2,6), (\alpha_2,\beta_2)=(3,0)$.}
    \label{fig:corner2}
\end{figure}

Without the loss of generality, we can further assume $\alpha_1 \ge \alpha_2$.

One checks that
\begin{align*}
	&\mathcal{I}^+ \setminus \delta \mathcal{I}^+ = \{a < \alpha_1\}\\
	&\delta \mathcal{I}^+ = 
	\{a= \alpha_1\}\\
	&\mathcal{O}^+ \cap \mathcal{O}^-= 	
	\{\alpha_1< a\le  \alpha_2\}\\
	&\delta \mathcal{I}^-  = \{ a = \alpha_2+1\}\\
	&\mathcal{I}^- \setminus \delta \mathcal{I}^- = \{\alpha_2 +1< a\}
\end{align*}
See Figure \ref{fig:corner2} for an example.
Let $i=(\alpha_2,N)$, $j=(\alpha_1+1,1)$,  $k_1=(\alpha_1,N)$, $k_2=(\alpha_1,1)$, $l_1=(\alpha_2+1,1)$, and $l_2=(\alpha_2+1,N)$.
We have the following observations
\begin{enumerate}[label=\roman*]
\item\label{case 3: lower left} $(sw\triangleright i, \mathcal E, \mathcal N )$ is $\mathrm{FC}(1,0)$,
\item\label{case 3: upper right} $(ne\triangleright j, \mathcal W, \mathcal S )$ is $\mathrm{FC}(-1,0)$,
\item\label{case 3: left} $(sw \triangleright k_1, \mathcal S, \mathcal E)$  and $(nw,k_2, \mathcal N, \mathcal E)$ are $\mathrm{BC}(-1, 0)$ with BC direction $\mathcal{W}$,
\item\label{case 3: right} $(ne \triangleright l_1, \mathcal W, \mathcal S)$  and $(se,l_2, \mathcal N, \mathcal W)$ are $\mathrm{BC}(-1, 0)$ with BC direction $\mathcal{E}$.
\end{enumerate}


If $(a,b) \in\mathcal{I}^+ \setminus \delta \mathcal{I}^+$, then by \ref{case 3: left}, then $ m^n_{\mathcal{E}(a-1,b)} 
= m^n_{\mathcal{N}(a,b-1)}
= m^n_{\mathcal{W}(a+1,b)}
= m^n_{\mathcal{S}(a,b+1)}=1
$
and so $\hat o^n_{(a,b)} = 4$. 

If $(a,b) \in\delta \mathcal{I}^+ $, then by \ref{case 3: left}, we have $m^n_{\mathcal{E}(a-1,b)} 
= m^n_{\mathcal{N}(a,b-1)}
= m^n_{\mathcal{S}(a,b+1)}=1$ and by \ref{case 3: upper right}, $m^n_{\mathcal{W}(a+1,b)} = -1$.Hence, $\hat o^n_{(a,b)} = 2$. 

If $(a,b) \in \mathcal{O}^+ \cap \mathcal{O}^-$, then by \ref{case 3: lower left}, $m^n_{\mathcal{E}(a-1,b)} 
= m^n_{\mathcal{N}(a,b-1)} = 1$ and by \ref{case 3: upper right}, $m^n_{\mathcal{W}(a+1,b)}
= m^n_{\mathcal{S}(a,b+1)}=-1$. Hence, $\hat o^n_{(a,b)} = 0$. 

If $(a,b) \in\delta \mathcal{I}^-$, then by \ref{case 3: right}, $m^n_{\mathcal{N}(a,b-1)}
= m^n_{\mathcal{W}(a+1,b)}
= m^n_{\mathcal{S}(a,b+1)}=-1$ and by \ref{case 3: lower left}, $m^n_{\mathcal{E}(a-1,b)} =1$. Hence, $\hat o^n_{(a,b)} = -2$. 

If $(a,b) \in\mathcal{I}^- \setminus \delta \mathcal{I}^-$, then by \ref{case 3: right},  $ m^n_{\mathcal{E}(a-1,b)} 
= m^n_{\mathcal{N}(a,b-1)}
= m^n_{\mathcal{W}(a+1,b)}
= m^n_{\mathcal{S}(a,b+1)}=-1
$ and so $\hat o^n_{(a,b)} = -4$. 
}





\bibliographystyle{ieeetran}
\bibliography{references}

\begin{thebibliography}{1}
\providecommand{\url}[1]{#1}
\csname url@samestyle\endcsname
\providecommand{\newblock}{\relax}
\providecommand{\bibinfo}[2]{#2}
\providecommand{\BIBentrySTDinterwordspacing}{\spaceskip=0pt\relax}
\providecommand{\BIBentryALTinterwordstretchfactor}{4}
\providecommand{\BIBentryALTinterwordspacing}{\spaceskip=\fontdimen2\font plus
\BIBentryALTinterwordstretchfactor\fontdimen3\font minus
  \fontdimen4\font\relax}
\providecommand{\BIBforeignlanguage}[2]{{%
\expandafter\ifx\csname l@#1\endcsname\relax
\typeout{** WARNING: IEEEtran.bst: No hyphenation pattern has been}%
\typeout{** loaded for the language `#1'. Using the pattern for}%
\typeout{** the default language instead.}%
\else
\language=\csname l@#1\endcsname
\fi
#2}}
\providecommand{\BIBdecl}{\relax}
\BIBdecl

\bibitem{weiss2000correctness}
Y.~Weiss, ``Correctness of local probability propagation in graphical models
  with loops,'' \emph{Neural computation}, vol.~12, no.~1, pp. 1--41, 2000.

\bibitem{weiss2001optimality}
Y.~Weiss and W.~T. Freeman, ``On the optimality of solutions of the max-product
  belief-propagation algorithm in arbitrary graphs,'' \emph{IEEE Transactions
  on Information Theory}, vol.~47, no.~2, pp. 736--744, 2001.

\bibitem{wainwright2004tree}
M.~Wainwright, T.~Jaakkola, and A.~Willsky, ``Tree consistency and bounds on
  the performance of the max-product algorithm and its generalizations,''
  \emph{Statistics and computing}, vol.~14, no.~2, pp. 143--166, 2004.

\bibitem{bayati2008max}
M.~Bayati, D.~Shah, and M.~Sharma, ``Max-product for maximum weight matching:
  Convergence, correctness, and lp duality,'' \emph{IEEE Transactions on
  Information Theory}, vol.~54, no.~3, pp. 1241--1251, 2008.

\bibitem{baxter2007exactly}
R.~J. Baxter, \emph{Exactly solved models in statistical mechanics}.\hskip 1em
  plus 0.5em minus 0.4em\relax Courier Corporation, 2007.

\bibitem{reyes2014lossy}
M.~G. Reyes, D.~L. Neuhoff, and T.~N. Pappas, ``Lossy cutset coding of bilevel
  images based on markov random fields,'' \emph{IEEE Transactions on Image
  Processing}, vol.~23, no.~4, pp. 1652--1665, 2014.

\bibitem{farmer2011cutset}
A.~Farmer, A.~Josan, M.~A. Prelee, D.~L. Neuhoff, and T.~N. Pappas, ``Cutset
  sampling and reconstruction of images,'' in \emph{2011 18th IEEE
  International Conference on Image Processing}.\hskip 1em plus 0.5em minus
  0.4em\relax IEEE, 2011, pp. 1909--1912.

\bibitem{reyes2011cutset}
M.~G. Reyes, ``Cutset based processing and compression of markov random
  fields,'' Ph.D. dissertation, The University of Michigan, 2011.

\bibitem{pearl2014probabilistic}
J.~Pearl, \emph{Probabilistic reasoning in intelligent systems: networks of
  plausible inference}.\hskip 1em plus 0.5em minus 0.4em\relax Morgan Kaufmann,
  2014.

\end{thebibliography}

%
%
%
%
%
%
%
%
%
%
%
%
%
%
%
%
%
%

\end{document}